\documentclass[11pt,a4paper]{article}
\usepackage{amsmath,amssymb,amsthm,graphicx,braket,multirow,authblk,amsthm,dcolumn,latexsym,subcaption,mathtools,cite,cancel}
\usepackage[normalem]{ulem}
\usepackage[a4paper]{geometry}
\DeclareSymbolFontAlphabet{\amsmathbb}{AMSb}%

\usepackage[unicode=true,
bookmarks=true,bookmarksopen=false,
breaklinks=false,pdfborder={0 0 0},colorlinks=true]
{hyperref}
\usepackage{xcolor}
\definecolor{cblue}{rgb}{0.16, 0.32, 0.75}
\definecolor{cred}{rgb}{0.7, 0.11, 0.11}
\hypersetup{%
	,linkcolor=cred
	,citecolor=cblue
	,urlcolor=black
}

\newtheorem{proposition}{Proposition}[section]

\newtheorem{lemma}{Lemma}[section]
\theoremstyle{remark}

\def\oper{{\mathchoice{\rm 1\mskip-4mu l}{\rm 1\mskip-4mu l}
		{\rm 1\mskip-4.5mu l}{\rm 1\mskip-5mu l}}}
\newcommand{\tr}{\operatorname{Tr}}
\newcommand{\ketbra}[2]{| #1 \rangle\!\langle #2 | }
\renewcommand{\Re}{\mathop{\mathrm{Re}}}
\renewcommand{\Im}{\mathop{\mathrm{Im}}}
\newcommand{\hilb}{\mathcal{H}}

\renewcommand{\i}{\mathrm{i}}
\newcommand{\e}{\mathrm{e}}

\newcommand{\m}{\mathrm{m}}

\DeclareMathAlphabet\mathbfcal{OMS}{cmsy}{b}{n}

\begin{document}	
	\title{\textbf{On the classicality of quantum dephasing processes}}
	
	\author[$\hspace{0cm}$]{Davide Lonigro$^{1,2,}$\footnote{davide.lonigro@ba.infn.it}}
	\affil[$1$]{\small Dipartimento di Fisica and MECENAS, Universit\`{a} di Bari, I-70126 Bari, Italy}
	\affil[$2$]{\small INFN, Sezione di Bari, I-70126 Bari, Italy}
	
	\author[$\hspace{0cm}$]{Dariusz Chru\'sci\'nski$^{3,}$\footnote{darch@fizyka.umk.pl}}
	\affil[$3$]{\small Institute of Physics, Faculty of Physics, Astronomy and Informatics, Nicolaus Copernicus University, Grudziadzka 5/7, 87-100 Toru\'n, Poland}
	
	\maketitle
	\vspace{-0.5cm}	
	
	\begin{abstract}
		We analyze the multitime statistics associated with pure dephasing systems repeatedly probed with sharp measurements, and search for measurement protocols whose statistics satisfies the Kolmogorov consistency conditions possibly up to a finite order. We find a rich phenomenology of quantum dephasing processes which can be interpreted in classical terms. In particular, if the underlying dephasing process is Markovian, we find sufficient conditions under which classicality at every order can be found: this can be reached by choosing the dephasing and measurement basis to be fully compatible or fully incompatible, that is, mutually unbiased bases (MUBs). For non-Markovian processes, classicality can only be proven in the fully compatible case, thus revealing a key difference between Markovian and non-Markovian pure dephasing processes.
	\end{abstract}
	
	\maketitle
	
	\section{Introduction}\label{sec:1}
	
	Quantum processes are fundamentally different from classical ones: therefore, a proper understanding of the classical--quantum boundary, involving the identification of the \textit{genuinely quantum} features of physical systems---those which cannot be simulated and/or reproduced by any system satisfying the laws of classical physics---is of particular importance~\cite{Zurek,Schloss,DEC2,Paterek,HHHH}. Aside from its intrinsic theoretical interest, this boundary also has practical implications for quantum technologies: taking advantage of the underlying quantum nature of the natural world, e.g.~quantum correlations~\cite{Paterek} and coherence~\cite{Zurek,Plenio}, is the key of diverse technological proposals such as quantum teleportation, cryptography or computation~\cite{QIT}. A groundbreaking example is provided by the violation of Bell inequalities by quantum correlations~\cite{Bell}; among many other such examples, we may point out quantum states of light corresponding to nonpositive values of their Wigner function (or singular Glauber--Sudarshan P-representation)~\cite{Zoller,Q-Optics}.
	
	A thorough analysis of physical systems must involve both its free evolution as well as its behavior under external interventions; this information can be encoded by a multitime statistics of joint probability distributions (correlation functions). In this framework, it was proposed in~\cite{Andrea-1,Andrea-2} to define the classicality of quantum systems in terms of the classicality of its multitime statistics. Let us briefly revise this approach. We shall consider a quantum system living on a Hilbert space $\hilb_{\rm S}$ with finite dimension $\dim\hilb_{\rm S}=d$, interacting with an environment (or bath) associated with a possibly infinite-dimensional Hilbert space $\hilb_{\rm B}$; let $\textbf{H}$ the self-adjoint operator on $\hilb_{\rm S}\otimes\hilb_{\rm B}$ generating the global dynamics of the system and the environment---the Hamiltonian.\footnote{
		For simplicity we shall mostly focus on the case in which the latter is time-independent, thus generating a homogeneous unitary propagator $\textbf{U}_{t,s}=\e^{-\i(t-s)\textbf{H}}$ on $\hilb_{\rm S}\otimes\hilb_{\rm B}$ solving the corresponding Schr\"odinger equation; however, most of our discussion can be replicated without substantial differences for time-dependent Hamiltonians, provided that a (generally inhomogeneous) unitary propagator exists.} We will denote by $\textbf{U}_{t,t_0}$ the unitary propagator generated by $\textbf{H}$.
	Assuming that, at some initial time $t_0$, the global system is in an uncorrelated (product) state $\rho_{t_0}\otimes\varrho_{\rm B}\in\mathcal{B}(\hilb_{\rm S}\otimes\hilb_{\rm B})\simeq\mathcal{B}(\hilb_{\rm S})\otimes\mathcal{B}(\hilb_{\rm B})$, the map $\Lambda_{t,t_0}$ describing the reduced dynamics of the system reads, for any $t\geq t_0$,
	\begin{equation}
		\Lambda_{t,t_0}(\rho_{t_0})=\tr_{\rm B}\,\mathbfcal{U}_{t,t_0}(\rho_{t_0}\otimes\varrho_{\rm B}),\qquad\mathbfcal{U}_{t,t_0}=\textbf{U}_{t,t_0}(\cdot)\textbf{U}_{t,t_0}^\dag,
	\end{equation}
	this map being, by construction, completely positive and trace preserving (CPTP). The map is said to be CP-divisible if it satisfies the additional property $\Lambda_{t_2,t_0}=\Lambda_{t_2,t_1}\Lambda_{t_1,t_0}$ and $\Lambda_{t_2,t_1}$ is CPTP for all $t_2\geq t_1\geq t_0$.
	
	We shall consider the idealized scenario in which such a system is repeatedly probed via sharp, identical and instantaneous measurements at times $t_1,\dots,t_n$, and evolves freely between each couple of consecutive measurements. This procedure brings about a family of time-continuous joint probability distributions (multitime statistics) defined as follows. Letting $\{P_x\}_{x}\subset\mathcal{B}(\hilb_{\rm S})$ be a projection-valued measure (PVM) representing the measurement apparatus, and $\mathcal{P}_x=P_x(\cdot)P_x$, we define, for all $n\in\mathbb{N}$,
	\begin{equation}\label{eq:quantumprob}
		\mathbb{P}_n(x_n,t_n;x_{n-1},t_{n-1};\ldots;x_1,t_1)=\tr\left[\left(\mathcal{P}_{x_n}\otimes\mathcal{I}_{\rm B}\right)\mathbfcal{U}_{t_n,t_{n-1}}\cdots\left(\mathcal{P}_{x_1}\otimes\mathcal{I}_{\rm B}\right)\mathbfcal{U}_{t_1,t_0}(\rho_{t_0}\otimes\varrho_{\rm B})\right],
	\end{equation}
	with $\mathcal{I}_{\rm B}$ being the identity map on $\hilb_{\rm B}$. This is a legitimate family of probability distributions formally analogous to the one associated with a classical stochastic process. However, while in the classical case the Kolmogorov consistency condition,
	\begin{equation}\label{eq:consistency}
		\mathbb{P}_{n-1}\left(x_n,t_n;\dots;\cancel{x_j,t_j};\dots;x_1,t_1\right)=\sum_{x_j=0}^{d-1}\mathbb{P}_n(x_n,t_n;\dots;x_j,t_j;\dots;x_1,t_1),
	\end{equation}
	holds, the statistics associated with a quantum system as described above will generally \textit{violate} this condition. In this regard, notice that Eq.~\eqref{eq:consistency} essentially means that not performing a measurement at the time $t_j$ is equivalent to performing a measurement at the same time and then forgetting about its outcome---in other words, measurements do not affect the system, which is generally not true in the quantum realm. On the other hand, by the Kolmogorov extension theorem~\cite{Kolmogorov}, Eq.~\eqref{eq:consistency} is satisfied if and only if there exist some \textit{classical} stochastic process reproducing the full multitime statistics; in this case the results of the experiment may be consistently explained in the framework of classical physics, whence no genuinely quantum feature is exhibited by the system.
	
	Because of that, following~\cite{Andrea-1,Andrea-2}, we shall adopt this definition: the process defined by the global Hamiltonian $\textbf{H}$, the initial state $\rho_{t_0}\otimes\varrho_{\rm B}$, and the family of projectors $\{P_x\}_{x}$ is
	\begin{itemize}
		\item $N$-\textit{classical}, if Eq.~\eqref{eq:consistency} is satisfied for all $n=1,\dots,N$.		
		\item \textit{classical}, if Eq.~\eqref{eq:consistency} is satisfied for all $n\in\mathbb{N}$, i.e.~the process is $N$-classical for all $N$.	
	\end{itemize}
	Interestingly, a quantum process with classical multitime statistics always satisfies the Leggett--Garg inequalities (LGI)~\cite{LG0,LG1,LG2}, which can be interpreted as a temporal analog of Bell inequalities; a violation of LGI necessarily implies that the underlying process is genuinely quantum.
	
	For an $N$-classical process, the intrinsically quantum nature of the system will only possibly emerge after $N+1$ measurements in the prescribed PVM, but will be effectively ``hidden" otherwise---that is, $N+1$ is the minimal number of measurements that must be performed on the system in order to see a genuinely quantum behavior. Most importantly, all these definitions are strictly dependent on both the measurement PVM, the state $\rho_{t_0}$ in which the system is initially prepared, and the initial state $\varrho_{\rm B}$ of the environment: even assuming the latter to be fixed, a system exhibiting ($N$--)classical behavior when measured with a certain apparatus and/or initially prepared in a given state may ``reveal" its quantum nature when measured or prepared differently.
	
	A particularly simple scenario is observed when the multitime statistics satisfies the Markov property:
	\begin{equation}
		\mathbb{P}(x_n,t_n|x_{n-1},t_{n-1};\ldots;x_1,t_1)=\mathbb{P}(x_n,t_n|x_{n-1},t_{n-1}),
	\end{equation}
	with $\mathbb{P}(\,\cdot\,|\,\cdot\,)$ corresponding to the conditional probability. In this special case, the whole multitime statistics of the process can be entirely reconstructed by the $1$-time distribution $\mathbb{P}_1(x_1,t_1)$ and the conditional probability $\mathbb{P}(x_2,t_2|x_1,t_1)$, whence the full (infinite) hierarchy of consistency conditions~\eqref{eq:consistency} ends up reducing to the Chapman--Kolmogorov equations, which can be written as
	\begin{eqnarray}\label{eq:consistency_markov1}
		\mathbb{P}_1(x_2,t_2)&=&\sum_{x=0}^{d-1}\mathbb{P}_2(x_2,t_2;x,t),\qquad\qquad\;\;\,t_2\geq t\geq t_0;\\
		\label{eq:consistency_markov2}
		\mathbb{P}_2(x_2,t_2;x_1,t_1)&=&\sum_{x=0}^{d-1}\mathbb{P}_3(x_2,t_2;x,t;x_1,t_1),\qquad t_2\geq t\geq t_1\geq t_0,
	\end{eqnarray}
	thus reducing to only two consistency conditions; in particular, in the Markovian case, $3$-classicality is sufficient for (hence equivalent to) classicality. It is also worth noting that the validity of the Markov property is ensured whenever the \textit{regression formula} holds~\cite{Lax,Zoller}, that is, when the full multitime statistics~\eqref{eq:quantumprob} induced by the system can be expressed in terms of the reduced dynamics alone:
	\begin{equation}\label{eq:quantumprob_regr}
		\mathbb{P}_n(x_n,t_n;x_{n-1},t_{n-1};\ldots;x_1,t_1)=\tr\left[\mathcal{P}_{x_n}\Lambda_{t_n,t_{n-1}}\cdots\mathcal{P}_{x_1}\Lambda_{t_1,t_0}(\rho_{t_0})\right];
	\end{equation}
	conversely, if the statistics is Markovian \textit{and} Eq.~\eqref{eq:quantumprob_regr} holds for $n=1,2$, then Eq.~\eqref{eq:quantumprob_regr} holds for all $n$.
	
	While adopting in this paper these mathematical definitions of Markovianity and classicality, it should be stressed that both physical concepts have been variously identified with diverse, generally inequivalent, mathematical properties. We refer to the reviews \cite{NM2,Piilo-I,Piilo-II} for a general discussion about quantum non-Markovianity; various definitions of quantum Markovianity and the intricate relations between them are presented in~\cite{NM1,NM4}, while we refer to~\cite{NM3} for overview of various approaches to describe the dynamics of non-Markovian open systems, and~\cite{PR} for the mathematical and physical properties of non-Markovian dynamical maps.
	
	The definition adopted in this paper is closely related to the mathematical formulation of quantum Markov stochastic processes proposed in~\cite{QP1,QP2,QP3} (cf.~also the review~\cite{Modi-PRX}) and closely related to the recent approach to quantum Markovianity proposed in~\cite{kavan1,kavan2,kavan3}, where the Markovianity of the corresponding process is characterized in terms of the factorization of the so-called quantum process tensor of the underlying quantum system, this property in turn being equivalent to the validity of quantum regression~\cite{NM4}. For a discussion of Markovianity based on the quantum regression formula see also~\cite{Francesco,Bassano}; exact results on the validity of quantum regression in simple models were found in~\cite{Davide-1,Davide-2,Samaneh}.			
	
	Similarly, many other approaches to the problem of classicality of quantum systems can be found in the literature. For example, in~\cite{Nori-PRL} the quantum evolution of an open system represented by a dynamical map $\Lambda_{t,s}$ is said to be classical it can be simulated by an ensemble of Hamiltonians $\{p_k,H_k\}$ on $\hilb_{\rm S}$, with $\{p_k\}_k$ being a probability distribution, such that, for any system state $\rho_{t_0}$, one has
	\begin{equation}
		\Lambda_{t,t_0}(\rho_{t_0}) = \sum_k p_k \,U^k_{t,t_0} \rho_{t_0} \big(U^k_{t,t_0}\big)^\dag,
	\end{equation}
	with $U^k_{t,t_0}=\e^{-\i(t-t_0) H_k}$: in this case, the (non-unitary) evolution of the system can be reproduced via a purely classical averaging procedure over distinct unitary evolutions. The above representation is often applied for disordered quantum systems described by Hamiltonian ensembles~\cite{Clemens-1,Clemens-2,Clemens-3,Clemens-4}. Note that such a concept of classicality immediately implies that the dynamics has to be unital, i.e.~$\Lambda_{t,t_0}(\oper) = \oper$ for all $t \geq t_0$ (cf.~also~\cite{NO}). With this approach, in~\cite{Nori} the (non)classicality of pure dephasing processes is analyzed for qubit systems, and a corresponding nonclassicality measure is proposed.
	
	In this paper we shall investigate the multitime statistics for a quantum system undergoing a pure dephasing evolution, focusing on the existence of particular choices of the initial state $\rho_{t_0}$ and the measurement basis for which classicality (possibly up to a finite number of measurements) is achieved, thus effectively ``hiding'' its quantum nature.
	This problem was recently analyzed in~\cite{Lukasz}, albeit from a slightly different perspective: the authors consider the scenario in which the system, after each measurement, is reset to the initial state $\rho_{t_0}$. We shall return to this point in Section~\ref{sec:5}.
	Here we will first find nontrivial dephasing processes that are $2$-classical, and then focus on the Markovian scenario, where necessary and sufficient conditions for the classicality (at any order) of a dephasing process are found. Interestingly, excepting particular cases, classicality can be achieved when the dephasing basis and the measurement one are either fully compatible or fully incompatible---that is, in the latter case, \textit{mutually unbiased bases} (MUBs)~\cite{MUB1} (cf.~also the review~\cite{MUB2}), possibly hinting at a deeper relation between the two concepts.
	
	The paper is organized as follows. In Section~\ref{sec:2} we provide a brief introduction about dephasing processes, in particular recalling a simple characterization of Markovian dephasing processes. In Section~\ref{sec:3} we show that a system initially prepared in a state which is diagonal in the dephasing basis, when measured with a proper class of unbiased bases with respect to the dephasing one, is $2$-classical but generally not $3$-classical; an analogous result is then shown to hold, under some conditions on the environment, for an arbitrary measurement basis if the initial state is the maximally mixed one. In Section~\ref{sec:4} we finally focus on the Markovian scenario and find conditions under which, in this case, the system is in fact classical when measured with each of these unbiased bases. Final considerations are outlined in Section~\ref{sec:5}.
	
	\section{Dephasing-type systems}\label{sec:2}
	
	Reprising the notation of Section~\ref{sec:1}, we shall consider global Hamiltonians on $\hilb_{\rm S}\otimes\hilb_{\rm B}$ in the following form. Given a PVM $\{E_j\}_{j}\subset\hilb_{\rm S}$ on the Hilbert space of the system, define
	\begin{equation}\label{eq:deph_pvm}
		\textbf{H}=\sum_{j}E_j\otimes H_j,
	\end{equation}
	with $\{H_j\}_j$ being a family of self-adjoint operators on $\hilb_{\rm B}$, whence $\textbf{H}$ is self-adjoint itself. Suppose that, at the initial time $t_0$, the global state of the system and the environment is a product state $\rho_{t_0}\otimes\varrho_{\rm B}$, and define
	\begin{eqnarray}\label{eq:ujl}
		\mathcal{U}_{t,s}^{j,\ell}&=&U_{t,s}^{j}(\cdot)\left(U_{t,s}^{\ell}\right)^\dag,\\
		\mathcal{E}^{j,\ell}&=&E_j(\cdot)E_\ell,
	\end{eqnarray}
	with $U_{t,s}^{j}$ being the unitary propagator associated with $H_j$. The map $\Lambda_{t,t_0}$ describing the reduced dynamics of the system is then given by
	\begin{eqnarray}\label{eq:reduced}
		\Lambda_{t,t_0}(\rho_{t_0})&=&\sum_{j,\ell}\tr\left[\mathcal{U}^{j,\ell}_{t,t_0}(\varrho_{\rm B})\right]\mathcal{E}^{j,\ell}(\rho_{t_0}),
	\end{eqnarray}
	thus being dependent on a matrix-valued function which we shall refer to as the \textit{dephasing matrix} of the process, crucially depending on the interplay between the block Hamiltonians $\{H_j\}_j$ and the environment state $\varrho_{\rm B}$. This is essentially a slight generalization of a typical dephasing channel, since at this stage we do not assume the projectors $E_j$ to be rank-one.
	
	The multitime statistics describing repeated measurements with a given PVM $\{P_x\}_{x}$ can be readily computed: for every $n$, the joint probability distribution reads
	\begin{equation}\label{eq:multi_dephasing}
		\mathbb{P}_n(x_n,t_n;\ldots;x_1,t_1)=\sum_{j_n,l_n}\cdots\sum_{j_1,l_1}\tr\left[\mathcal{P}_{x_n}\mathcal{E}^{j_n,\ell_n}\cdots\mathcal{P}_{x_1}\mathcal{E}^{j_1,\ell_1}(\rho_{t_0})\right]\tr\left[\mathcal{U}^{j_n,\ell_n}_{t_n,t_{n-1}}\cdots\mathcal{U}^{j_1,\ell_1}_{t_1,t_0}(\varrho_{\rm B})\right],
	\end{equation}
	again with $\mathcal{P}_x=P_x(\cdot)P_x$. Again, each term of the sum decomposes into the product of a term containing all information about the preparation--measurement process, and another encoding all information about the environment. We may refer to the latter as the \textit{dephasing tensor}. In particular, by a direct check one obtains the following result:
	\begin{proposition}[\!\!{\cite{Davide-2}}]\label{prop:markov}
		The dephasing system satisfies the regression equality~\eqref{eq:quantumprob} if and only if
		\begin{equation}\label{eq:dephasing_regr}
			\tr\left[\mathcal{U}^{j_n,\ell_n}_{t_n,t_{n-1}}\cdots\mathcal{U}^{j_1,\ell_1}_{t_1,t_0}(\varrho_{\rm B})\right]=\prod_{k=1}^n\tr\left[\mathcal{U}^{j_k,\ell_k}_{t_k,t_{k-1}}(\varrho_{\rm B})\right],
		\end{equation}
	with $\mathcal{U}^{j,\ell}_{t,s}$ as in Eq. \eqref{eq:ujl}.
	\end{proposition}
	While this claim is proven in~\cite{Davide-2} for a standard dephasing channel (i.e., all projectors $E_j$ being rank-one), its extension to the general case is immediate. Eq.~\eqref{eq:dephasing_regr} means that, for dephasing processes, the validity of the regression formula---and hence Markovianity---reduces to a factorization property of the $n$-time dephasing tensor in terms of $2$-time tensors, i.e.~dephasing matrices. This condition is clearly independent of the particular preparation--measurement protocol: the Markovianity of any dephasing process only depends on the initial environment state $\varrho_{\rm B}$ and the Hamiltonian $\textbf{H}$, without any dependence on either the initial system state $\rho_{t_0}$ nor the particular choice of measurement PVM. In this sense, Prop.~\ref{prop:markov} provides a complete characterization of (non)Markovianity in dephasing processes. As we will see, the situation is much more involved for classicality, in which both the initial state $\rho_{t_0}$ and the measurement PVM play a fundamental role.
	
	The hierarchy of conditions~\eqref{eq:dephasing_regr} has two important consequences. First of all, it contains the following condition
	\begin{equation}
		\tr\left[\mathcal{U}^{j,\ell}_{t_2,t_0}(\varrho_{\rm B})\right]=\tr\left[\mathcal{U}^{j,\ell}_{t_2,t_1}(\varrho_{\rm B})\right]\tr\left[\mathcal{U}^{j,\ell}_{t_1,t_0}(\varrho_{\rm B})\right],
	\end{equation}
	which, for time-independent Hamiltonians, implies, for some $\epsilon_{j\ell}\in\mathbb{R}$ and $\gamma_{j\ell}\geq0$,
	\begin{equation}
		\tr\left[\mathcal{U}^{j,\ell}_{t,s}(\varrho_{\rm B})\right]=\e^{-\left(\i\epsilon_{j\ell}+\frac{1}{2}\gamma_{j\ell}\right)(t-s)},\qquad t\geq s\geq t_0,
	\end{equation}
	and clearly means that the reduced dynamics induced by $\textbf{H}$ must satisfy the semigroup property, whence $\Lambda_{t,s}$ must be a GKLS semigroup~\cite{GKS,L}. However, keep in mind that this is only a necessary condition: counterexamples of dephasing semigroups blatantly violating regression can be easily constructed. However, while very restrictive, this condition is satisfied in (at least) one case of physical interest: by choosing $\textbf{H}$ as the dephasing-type spin--boson model with flat form factor, and $\varrho_{\rm B}$ as the vacuum state of the boson field, the couple $(\textbf{H},\varrho_{\rm B})$ does indeed satisfy the full hierarchy~\eqref{eq:dephasing_regr}, as shown in~\cite{Davide-2}.
	
	Furthermore, Eq.~\eqref{eq:dephasing_regr} cannot hold in the \textit{commutative} case---that is, when the Hamiltonians $\{H_j\}_j$ form a commutative family\footnote{Since we are dealing with possibly unbounded operators, this is to be interpreted in the following sense: $\e^{-\i tH_j}\e^{-\i sH_\ell}=\e^{-\i sH_\ell}\e^{-\i tH_j}$ for all $j,\ell=0,\dots,d-1$ and all $t,s\in\mathbb{R}$}---unless the dephasing is trivial, that is, all elements of the dephasing matrix have unit modulus,
	\begin{equation}
		\left|\tr\left[\mathcal{U}^{j,\ell}_{t_2,t_1}(\varrho_{\rm B})\right]\right|^2=1,
	\end{equation}
	for all values of the parameters: the reduced dynamics induced on $\rho_{t_0}$ (cf.~Eq.~\eqref{eq:reduced}) only involves its off-diagonal elements acquiring a time-dependent phase term, their modulus being unchanged~\cite{Davide-2}.
	
	To conclude, for future convenience let us list here some elementary properties of the dephasing matrix and tensor.
	\begin{proposition}\label{prop:markommute}
		For every $n$ and every value of the indices, we have
		\begin{equation}
			\tr\left[\mathcal{U}^{j_n,j_n}_{t_n,t_{n-1}}\mathcal{U}^{j_{n-1},\ell_{n-1}}_{t_{n-1},t_{n-2}}\cdots\mathcal{U}^{j_1,\ell_1}_{t_1,t_0}(\varrho_{\rm B})\right]=\tr\left[\mathcal{U}^{j_{n-1},\ell_{n-1}}_{t_{n-1},t_{n-2}}\cdots\mathcal{U}^{j_1,\ell_1}_{t_1,t_0}(\varrho_{\rm B})\right];
		\end{equation}
		besides, in the case in which the dephasing is either Markovian or generated by commuting Hamiltonians, for every $k=1,\dots,n$ we have
		\begin{eqnarray}
			&&	\tr\left[\mathcal{U}^{j_n,\ell_n}_{t_n,t_{n-1}}\cdots\mathcal{U}^{j_{k+1},\ell_{k+1}}_{t_{k+1},t_{k}}\mathcal{U}^{j_{k},j_{k}}_{t_{k},t_{k-1}}\mathcal{U}^{j_{k-1},\ell_{k-1}}_{t_{k-1},t_{k-2}}\cdots\mathcal{U}^{j_1,\ell_1}_{t_1,t_0}(\varrho_{\rm B})\right]\nonumber\\&=&	\tr\left[\mathcal{U}^{j_n,\ell_n}_{t_n,t_{n-1}}\cdots\mathcal{U}^{j_{k+1},\ell_{k+1}}_{t_{k+1},t_{k}}\mathcal{U}^{j_{k-1},\ell_{k-1}}_{t_{k-1},t_{k-2}}\cdots\mathcal{U}^{j_1,\ell_1}_{t_1,t_0}(\varrho_{\rm B})\right].
		\end{eqnarray}
	\end{proposition}
	That is: in the general case, all entries of the $n$-time dephasing tensor characterized by the $n$th couple of indices $(j_n,\ell_n)$ having equal values can be evaluated by substituting the unitary map $\mathcal{U}^{(j_n,j_n)}_{t_n,t_{n-1}}$ with the identity map; furthermore, if either Eq.~\eqref{eq:dephasing_regr} or the block Hamiltonians form a commuting family, the same holds for \textit{every} other pair $(j_k,\ell_k)$. As we will see, the latter property will greatly simplify the calculation of multitime probabilities in the particular cases we are interested in.
	
	\section{2-classical dephasing processes}\label{sec:3}	
	
	We shall now search for particular preparation--measurement protocols which yield a classical or $N$-classical process. This will ultimately depend on the interplay between the two PVMs that enter the process:
	\begin{itemize}
		\item the dephasing PVM $\{E_j\}_{j}$ in the definition~\eqref{eq:deph_pvm} of the global Hamiltonian;
		\item the measurement PVM $\{P_x\}_{x}$ with which the system is repeatedly measured,
	\end{itemize}
	but in general may---and will---also depend nontrivially on the properties of the environment.
	
	Clearly, if the two PVMs coincide, then the process is indeed classical---and shows a trivial dependence on time. Indeed, in this case we must set $\mathcal{P}_{x_n}=\mathcal{E}^{x_n,x_n}$ in Eq.~\eqref{eq:multi_dephasing}, whence
	\begin{equation}
		\mathbb{P}_n(x_n,t_n;\ldots;x_1,t_1)=\delta_{x_n,x_{n-1}}\cdots\delta_{x_2,x_1}\tr[\mathcal{P}_{x_1}(\rho_{t_0})],
	\end{equation}
	thus obtaining an elementary process, independent of all times $t_1,\dots,t_n$, which is clearly classical: this is an obvious consequence of the fact that the measurement do not interfere at all with the dephasing. More generally, the same happens if the two PVMs are commuting. Clearly, a genuinely time-dependent (and generally nonclassical) process can only be obtained when the two PVMs do \textit{not} commute, that is, when the measurement is not \textit{fully compatible} with the dephasing.
	
	We may wonder whether there exist preparation--measurement protocols which yield time-dependent dephasing processes exhibiting a classical behavior, that is, satisfying the consistency condition~\eqref{eq:consistency} possibly up to some finite order, despite the fact that the system is measured with some PVM not fully compatible with its ``natural" one. We shall start in this section by analyzing some of these configurations which ensure $2$-classicality, and then proceed in Section~\ref{sec:4} to investigate the Markovian case, where classicality at any order can be in fact achieved.	
	
	\subsection{2-classicality in the qubit case}
	
	We shall start from the case $d=2$, that is, a qubit undergoing pure dephasing in some orthonormal basis $\{\ket{\e_0},\ket{\e_1}\}$ under the action of a global Hamiltonian $\textbf{H}$ in the form
	\begin{equation}\label{eq:h_qubit}
		\textbf{H}=\ketbra{\e_0}{\e_0}\otimes H_0+\ketbra{\e_1}{\e_1}\otimes H_1.
	\end{equation}
	Let us consider a measurement basis $\{\ket{\m_0(\phi;\theta)},\ket{\m_1(\phi;\theta)}\}\subset\hilb_{\rm S}$ in the form
	\begin{eqnarray}\label{eq:genbasis}
		\ket{\m_0(\phi;\theta)}&=&\cos \theta\ket{\e_0}+\e^{\i \phi} \sin \theta\ket{\e_1},\\
		\label{eq:genbasis2}
		\ket{\m_1(\phi;\theta)}&=&\sin \theta\ket{\e_0}-\e^{\i \phi} \cos \theta\ket{\e_1}.
	\end{eqnarray}
	for some angles $\phi,\theta$. Up to global phases which do not affect the multitime statistics, the vectors~\eqref{eq:genbasis}--\eqref{eq:genbasis2} span all possible measurements on the qubit; up to a flip of the two basis vectors, it suffices to consider $0\leq\theta\leq\pi/2$: in particular, the value $\theta=0$ yields the dephasing basis, while $\theta=\pi/4$ yields (setting $\ket{m_x(\phi,\pi/4)}\equiv\ket{m_x(\phi)}$)
	\begin{eqnarray}\label{eq:mubs_qubit}
		\ket{\m_0(\phi)}&=&\frac{1}{\sqrt{2}}\left(\ket{\e_0}+\e^{\i\phi}\ket{\e_1}\right),\\
		\label{eq:mubs_qubit2}
		\ket{\m_1(\phi)}&=&\frac{1}{\sqrt{2}}\left(\ket{\e_0}-\e^{\i\phi}\ket{\e_1}\right) .
	\end{eqnarray}
	Importantly, for any value of $\phi$ we have $|\braket{\m_x(\phi)|\e_j}|^2=1/2$ for all $x,j=1,2$: that is, the basis~\eqref{eq:mubs_qubit} and the dephasing basis are mutually unbiased bases (MUBs)~\cite{MUB1,MUB2}. In particular, the choices $\phi=0$ and $\phi=\pi/2$ correspond respectively to the eigenvectors of the Pauli operators $\sigma_x$ and $\sigma_y$, two choices of obvious practical relevance.
	
	Choosing the initial state to be diagonal in the dephasing basis:
	\begin{equation}\label{rho-p}
		\rho_{t_0}=p\ketbra{\e_0}{\e_0}+(1-p)\ketbra{\e_1}{\e_1},\qquad 0\leq p\leq 1,
	\end{equation}
	a long but straightforward calculation allows one to determine the $1$-time and $2$-time probability distributions associated to this preparation--measurement protocol. For any choice of $\theta$ and $\phi$, one obtains
	\begin{eqnarray}\label{q21}
		&&	\sum_{x_1=0,1}\mathbb{P}_2(x_2,t_2;x_1,t_1)-\mathbb{P}_1(x_2,t_2)\nonumber\\&=&(-1)^{x_2}\frac{1}{8} \sin 2 \theta\, \sin 4 \theta
		\Big(p
		\tr\left[\mathcal{U}^{0,1}_{t_2,t_1}\mathcal{U}^{0,0}_{t_1,t_0}(\varrho_{\rm B})\right]
		+p
		\tr\left[\mathcal{U}^{1,0}_{t_2,t_1}\mathcal{U}^{0,0}_{t_1,t_0}(\varrho_{\rm B})\right]\nonumber\\
		&&
		+(p-1)
		\tr\left[\mathcal{U}^{0,1}_{t_2,t_1}\mathcal{U}^{1,1}_{t_1,t_0}(\varrho_{\rm B})\right]
		+(p-1)
		\tr\left[\mathcal{U}^{1,0}_{t_2,t_1}\mathcal{U}^{1,1}_{t_1,t_0}(\varrho_{\rm B})\right]
		-2(2p-1)\Big).
	\end{eqnarray}
	The above result, which does not depend on $\phi$, corresponds to the product of a time-independent term $\propto\sin2\theta\sin4\theta$ encoding all information about the measurement, times a time-dependent term depending on the choice of initial state and the properties of the environment. $2$-classicality is obtained, regardless the properties of the environment, in the cases $\theta=0$ (fully compatible measurement) and $\theta=\pi/4$ (fully incompatible measurement). The calculation above clearly shows that, in the qubit case, there are no other possibilities to obtain $2$-classicality without making further assumptions on the initial state of the system (the value of $p$) or on the properties of the environment.
	\begin{proposition}\label{prop:qubit1}
		Let $d=2$; consider an initial state $\rho_{t_0}$ diagonal in the basis $\{\ket{\e_j}\}_{j=0,1}$, and any measurement basis in the form~\eqref{eq:genbasis}--\eqref{eq:genbasis2}. Then the process is $2$-classical if one of the following conditions hold:
		\begin{itemize}
			\item $\theta=0$, that is, the measurement basis is fully compatible with the dephasing one;
			\item $\theta=\pi/4$, that is, the measurement basis is fully incompatible (unbiased) with the dephasing one;
		\end{itemize}
	\end{proposition}
	In all other cases, the process will fail the $2$-classicality test: two measurements will be enough to reveal the nonclassicality of the process. In particular, for any fixed choice of the initial state, the deviations from classicality will be maximal when $\sin2\theta\sin4\theta$ is maximal: this happens for
	\begin{equation}
		\theta_{\rm max}^{(1)}=\frac{1}{2}\arctan\sqrt{2},\qquad\theta_{\rm max}^{(2)}=\frac{1}{2}\left(\pi-\arctan\sqrt{2}\right).
	\end{equation}
	
	We may wonder whether the cases listed in Prop.~\ref{prop:qubit1} are the \textit{only} possible ones in which $2$-classicality is achieved. An (almost) affirmative answer to this question can be obtained by making one of the following assumptions on the environment. Suppose that either the process is Markovian, or the Hamiltonians $H_0,H_1$ generating the environment dynamics, cf.~Eq.~\eqref{eq:h_qubit}, commute. In both cases, setting for brevity
	\begin{equation}
		\varphi_{t,s}:=\tr\,\mathcal{U}^{0,1}_{t,s}(\varrho_{\rm B}),
	\end{equation}
	by using Prop.~\ref{prop:markommute} Eq.~\eqref{q21} simplifies as such:
	\begin{equation}\label{q21bis}
		\sum_{x_1=0,1}\mathbb{P}_2(x_2,t_2;x_1,t_1)-\mathbb{P}_1(x_2,t_2)=(-1)^{x_2}\left(\frac{1}{2}-p\right) \sin 2 \theta\, \sin 4 \theta \,(1-\Re\varphi_{t_2,t_1}).
	\end{equation}
	Excluding the trivial case $\varphi_{t_2,t_1}\equiv1$, this quantity vanishes identically if either $\theta=0$, $\theta=\pi/4$, or $p=1/2$: interestingly, $2$-classicality is also achieved for a special case of initial state---the maximally mixed state---regardless the measurement procedure.
	\begin{proposition}\label{prop:qubit2}
		Let $d=2$; consider an initial state $\rho_{t_0}$ diagonal in the basis $\{\ket{\e_j}\}_{j=0,1}$, and any measurement basis in the form~\eqref{eq:genbasis}--\eqref{eq:genbasis2}. Suppose that either the dephasing process is Markovian or the Hamiltonians $H_0,H_1$ commute. Then the process is $2$-classical if and only if one of the following conditions holds:
		\begin{itemize}
			\item $\theta=0$, that is, the measurement basis is fully compatible with the dephasing one;
			\item $\theta=\pi/4$, that is, the measurement basis is fully incompatible (unbiased) with the dephasing one;
			\item $p=1/2$, that is, the system is prepared in the maximally mixed state.
		\end{itemize}
	\end{proposition}
	These results can be summarized as follows. Excluding the elementary case of a fully compatible measurement, preparation--measurement protocols yielding a genuinely time-dependent and $2$-classical dephasing process \textit{do} exist: one prepares the system in a state $\rho_{t_0}$ which is diagonal in the dephasing basis, and measures it in any of the unbiased bases~\eqref{eq:mubs_qubit}, thus obtaining a $2$-classical process regardless the properties of the environment. If the process is Markovian or the operators $\{H_0,H_1\}$ commute, this is indeed the only way to achieve a $2$-classical time-dependent process from an initial diagonal state---excepting the peculiar case in which the latter is maximally mixed. These results have a clear physical explanation: any diagonal state is left invariant by the dephasing dynamics, whence it evolves trivially, and is then mapped in the maximally mixed one when measured in any basis which is unbiased with respect to the dephasing one. No deviation from classicality will be thus observed in both cases.
	
	This discussion clearly shows the distinguished role played by MUBs in the classicality of qubit dephasing processes; in fact, we will now see that similar properties hold in the general case $d>2$.
	
	\subsection{2-classicality vs. MUBs}\label{subsec:mubs}
	
	Going beyond the qubit scenario, let us now consider a system of arbitrary dimension $d$. Suppose $E_j=\ketbra{\e_j}{\e_j}$ for some orthonormal basis $\{\ket{\e_j}\}_{j=0,\dots,d-1}\subset\hilb_{\rm S}$. We shall consider dephasing processes corresponding to a measurement basis chosen as follows: given an arbitrary vector of phases	
	\begin{equation}
		\boldsymbol\phi = (\phi_0,\dots,\phi_{d-1}),
	\end{equation}
	let us define 	
	\begin{equation}\label{eq:mubs}
		\ket{\m_x( \boldsymbol\phi)}=\frac{1}{d^{1/2}}\sum_{j=0}^{d-1}\omega^{jx}\e^{\i\phi_j}\ket{\e_j},\qquad\omega=\e^{2\pi\i/d}.
	\end{equation}
	For $\boldsymbol\phi=(0,\ldots,0)$, this is the discrete Fourier transform of the dephasing basis. Of course, since the vectors are uniquely associated to their projectors only up to a global phase, we have in fact $d-1$ free parameters (we may e.g.~set $\phi_0=0$). This expression generalizes Eqs.~\eqref{eq:mubs_qubit}--\eqref{eq:mubs_qubit2}; indeed, for any choices of the phases, $\{\ket{\e_j}\}_{j=0,\dots,d-1}$ and $\{\ket{\m_x( \boldsymbol\phi)}\}_{x=0,\dots,d-1}\subset\hilb_{\rm S}$ are MUBs:
	\begin{equation}
		|\!\braket{\m_x( \boldsymbol\phi)|\e_j}\!|^2=\frac{1}{d},\qquad j,x=0,\dots,d-1.
	\end{equation}
	By Eq.~\eqref{eq:multi_dephasing}, one easily computes the joint probability distributions associated with this protocol. 	For simplicity, hereafter we shall often abuse the notation by setting
	\begin{equation}
		\ket{\m_x( \boldsymbol\phi)}\equiv\ket{x},\qquad\ket{\e_j}\equiv\ket{j}.
	\end{equation}	
	In the cases $n=1,2$,
	\begin{eqnarray}\label{eq:prob1}
		\mathbb{P}_1(x_1,t_1)&=&\sum_{j_1,\ell_1}\tr\left[P_{x_1}E_{j_1}\rho_{t_0}E_{\ell_1}P_{x_1}\right]\tr\left[\mathcal{U}^{j_1,\ell_1}_{t_1,t_0}(\varrho_{\rm B})\right]\nonumber\\
		&=&\frac{1}{d}\sum_{j_1,\ell_1}\omega^{-(j_1-\ell_1)x_1}\braket{j_1|\rho_{t_0}|\ell_1}\e^{-\i(\phi_{j_1}-\phi_{\ell_1})}\tr\left[\mathcal{U}^{j_1,\ell_1}_{t_1,t_0}(\varrho_{\rm B})\right] ,
	\end{eqnarray}
	\begin{eqnarray}\label{eq:prob2}
		\mathbb{P}_2(x_2,t_2;x_1,t_1)&=&\sum_{j_2,\ell_2}\sum_{j_1,\ell_1}\tr\left[P_{x_2}E_{j_2}P_{x_1}E_{j_1}\rho_{t_0}E_{\ell_1}P_{x_1}E_{\ell_2}P_{x_2}\right]\tr\left[\mathcal{U}^{j_2,\ell_2}_{t_2,t_1}\mathcal{U}^{j_1,\ell_1}_{t_1,t_0}(\varrho_{\rm B})\right]\nonumber\\
		&=&\frac{1}{d^3}\sum_{j_2,\ell_2}\sum_{j_1,\ell_1}
		\omega^{-(j_2-\ell_2)x_2}	
		\omega^{[(j_2-\ell_2)-(j_1-\ell_1)]x_1}\braket{j_1|\rho_{t_0}|\ell_1}\nonumber\\&&\times\e^{-\i(\phi_{j_1}-\phi_{\ell_1})}\tr\left[\mathcal{U}^{j_2,\ell_2}_{t_2,t_1}\mathcal{U}^{j_1,\ell_1}_{t_1,t_0}(\varrho_{\rm B})\right],\nonumber\\
	\end{eqnarray}
	this structure being immediately generalizable to $n$-time probabilities, for example\small
	\begin{eqnarray}
		\mathbb{P}_3(x_3,t_3;x_2,t_2;x_1,t_1)&=&\frac{1}{d^5}\sum_{j_3,\ell_3}\sum_{j_2,\ell_2}\sum_{j_1,\ell_1}
		\omega^{-(j_3-\ell_3)x_3}\omega^{[(j_3-\ell_3)-(j_2-\ell_2)]x_2}	
		\omega^{[(j_2-\ell_2)-(j_1-\ell_1)]x_1}\braket{j_1|\rho_{t_0}|\ell_1}\nonumber\\&&\times\,\e^{-\i(\phi_{j_1}-\phi_{\ell_1})}\tr\left[\mathcal{U}^{j_3,\ell_3}_{t_3,t_2}\mathcal{U}^{j_2,\ell_2}_{t_2,t_1}\mathcal{U}^{j_1,\ell_1}_{t_1,t_0}(\varrho_{\rm B})\right] .
	\end{eqnarray}\normalsize
	Interestingly, in all cases only the sums on the indices $j_1$ and $\ell_1$ carry a dependence on the phases $\phi_0,\dots,\phi_{d-1}$ which appear in the definition~\eqref{eq:mubs} of the measurement basis.
	
	Now, at the level of $2$-time probability distributions, the only nontrivial consistency condition is the following:
	\begin{equation}
		\sum_{x_1}\mathbb{P}_2(x_2,t_2;x_1,t_1)=\mathbb{P}_1(x_2,t_2),\qquad t_2\geq t_1\geq t_0,
	\end{equation}
	since marginalizing over the last outcome always yields the $1$-time probability. On the other hand, by Eq.~\eqref{eq:prob2}, summing over $x_1$ we have\small
	\begin{eqnarray}\label{eq:prob2_marginalized}
		&&\sum_{x_1=0}^{d-1}\mathbb{P}_2(x_2,t_2;x_1,t_1)\nonumber\\&=&\frac{1}{d^3}\!\sum_{j_2,\ell_2}\sum_{j_1,\ell_1}
		\omega^{-(j_2-\ell_2)x_2}\!\!
		\left(\sum_{x_1=0}^{d-1}\omega^{[(j_2-\ell_2)-(j_1-\ell_1)]x_1}\right)\!\braket{j_1|\rho_{t_0}|\ell_1}\e^{-\i(\phi_{j_1}-\phi_{\ell_1})}\!\tr\!\left[\mathcal{U}^{j_2,\ell_2}_{t_2,t_1}\mathcal{U}^{j_1,\ell_1}_{t_1,t_0}(\varrho_{\rm B})\right],\nonumber\\
	\end{eqnarray}\normalsize
	and, in general, the quantities in Eqs.~\eqref{eq:prob1} and~\eqref{eq:prob2_marginalized} do not coincide. However, let us consider again the particular case in which the initial state $\rho_{t_0}$ is \textit{diagonal} in the dephasing basis. We will prove the following claim.
	
	\begin{proposition}\label{prop:2classicality_mubs}
		Consider an initial state $\rho_{t_0}$ diagonal in the basis $\{\ket{\e_j}\}_{j=0,\dots,d-1}$, and, given any family of phases $\phi_0,\dots,\phi_{d-1}$, choose $\{\ket{\m_x( \boldsymbol\phi)}\}_{x=0,\dots,d-1}$ as in Eq.~\eqref{eq:mubs} as measurement basis. Then the process is $2$-classical.
	\end{proposition}
	\begin{proof}
		Recalling that $\tr\rho_{t_0}=1$, and again using the shorthand $\ket{\e_j}\equiv\ket{j}$, we have
		\begin{eqnarray}
			\mathbb{P}_1(x_1,t_1)&=&\frac{1}{d}\sum_{j_1}\braket{j_1|\rho_{t_0}|j_1}=\frac{1}{d},
		\end{eqnarray}
		this quantity being independent of time---indeed, independent of any parameter of the process but the dimension of the Hilbert space: the evolution is trivial because of the initial state being diagonal in the dephasing basis, and all outcomes are equally probable at any time due to the fact that the measurement and the dephasing basis are MUBs. Besides, by Eq.~\eqref{eq:prob2_marginalized} one finds
		\begin{eqnarray}
			&&\sum_{x_1=0}^{d-1}\mathbb{P}_2(x_2,t_2;x_1,t_1)\nonumber\\
			&=&\frac{1}{d^2}\sum_{j_2,\ell_2}\sum_{j_1}
			\omega^{-(j_2-\ell_2)x_2}	
			\left(\sum_{x_1=0}^{d-1}\omega^{(j_2-\ell_2)x_1}\right)\braket{j_1|\rho_{t_0}|j_1}\tr\left[\mathcal{U}^{j_2,\ell_2}_{t_2,t_1}\mathcal{U}^{j_1,j_1}_{t_1,t_0}(\varrho_{\rm B})\right]\nonumber\\
			&=&\frac{1}{d^2}\sum_{j_2,\ell_2}\sum_{j_1}
			\omega^{-(j_2-\ell_2)x_2}	
			\delta_{j_2-\ell_2,0}\braket{j_1|\rho_{t_0}|j_1}\tr\left[\mathcal{U}^{j_2,\ell_2}_{t_2,t_1}\mathcal{U}^{j_1,j_1}_{t_1,t_0}(\varrho_{\rm B})\right]\nonumber\\
			&=&\frac{1}{d^2}\sum_{j_2}\sum_{j_1}\braket{j_1|\rho_{t_0}|j_1} =\frac{1}{d} ,
		\end{eqnarray}
		where we have used the fact that $\sum_{x_1}\omega^{(j_2-\ell_2)x_1}$ is either equal to $d$ if $j_2=\ell_2\!\!\mod d$, or $0$ otherwise; but, since both $j_2$ and $\ell_2$ range between $1$ and $d$, this only happens if $j_2=\ell_2$.
	\end{proof}
	This result generalizes Prop.~\ref{prop:qubit1} to the qu$d$it case: $2$-classicality holds for any preparation--measurement protocol involving a diagonal state $\rho_{t_0}$ in the dephasing basis, and a measurement with respect to any of the unbiased bases~\eqref{eq:mubs}. Like in the qubit case, it holds independently of the choice of the phases in the definition~\eqref{eq:mubs} of the unbiased basis with which the measurement is performed: in general, the fact that $\rho_{t_0}$ is diagonal in the dephasing basis erases all information about said phases from the full multitime statistics of the process. Besides, it holds regardless of the choice of the Hamiltonians $H_0,\dots,H_{d-1}$ as well as the environment state $\varrho_{\rm B}$. It may be conjectured that this is the only measurement choice yielding a genuinely time-dependent $2$-classical process---excepting the particular case in which $\rho_{t_0}$ is maximally mixed, cf.~the discussion in the next subsection.
	
	Finally, notice that in general the results of Prop.~\ref{prop:2classicality_mubs} (as well as its qubit version) cannot be improved from $2$-classicality to higher orders: a direct check shows indeed that, without further assumptions, the process is not $3$-classical. We have a full family of quantum processes which are $2$-classical but generally not classical: the quantum nature of the system will only emerge when taking into account the $3$-time probability distribution, i.e., by performing three measurements on the system. By doing so, the particular features of the environment will play a fundamental role, crucially determining the statistics of the system and generally disrupting classicality. In this sense, we may say that the process exhibits, in general, a hidden non-classicality, only accessible with a sufficient number of measurements.
	
	\subsection{2-classicality for maximally mixed initial state}\label{subsec:maxmix}
	
	In the qubit scenario, $2$-classicality was shown to hold, other than by choosing the dephasing and measurement bases to be fully compatible or incompatible, also by simply preparing the system in the maximally mixed state. As it turns out, this property has, again, a simple extension to the qu$d$it case.
	
	We shall consider the general setting in which both the dephasing and the measurement PVMs, $\{E_j\}_{j}$ and $\{P_x\}_x$, are arbitrary, and choose the initial state of the system as the maximally mixed one, i.e.~$\rho_{t_0}= \oper/d$. In this case,
	\begin{eqnarray}\label{eq:p1_maxmix}
		\mathbb{P}_1(x_1,t_1)&=&\frac{1}{d}\sum_{j_1,\ell_1}\tr\left[P_{x_1}E_{j_1}E_{\ell_1}P_{x_1}\right]\tr\left[\mathcal{U}^{j_1,\ell_1}_{t_1,t_0}(\varrho_{\rm B})\right]\nonumber\\
		&=&\frac{1}{d}\sum_{j_1}\tr\left[P_{x_1}E_{j_1}P_{x_1}\right] = \frac{1}{d}\tr[P_{x_1}],
	\end{eqnarray}
	where we used the properties $E_jE_\ell=\delta_{j\ell}E_j$ and $\sum_{j_1}E_{j_1}=\oper$. Again, the $1$-time probability is time-independent: the probability of each outcome is simply equal to the ratio between the rank of the projector and the dimension of the system space. Similarly,
	\begin{eqnarray}\label{eq:p2_maxmix}
		\mathbb{P}_2(x_2,t_2;x_1,t_1)&=&\frac{1}{d}\sum_{j_2,\ell_2}\sum_{j_1,\ell_1}\tr\left[P_{x_2}E_{j_2}P_{x_1}E_{j_1}E_{\ell_1}P_{x_1}E_{\ell_2}P_{x_2}\right]\tr\left[\mathcal{U}^{j_2,\ell_2}_{t_2,t_1}\mathcal{U}^{j_1,\ell_1}_{t_1,t_0}(\varrho_{\rm B})\right]\nonumber\\
		&=&\frac{1}{d}\sum_{j_2,\ell_2}\sum_{j_1}\tr\left[P_{x_2}E_{j_2}P_{x_1}E_{j_1}P_{x_1}E_{\ell_2}P_{x_2}\right]\tr\left[\mathcal{U}^{j_2,\ell_2}_{t_2,t_1}\mathcal{U}^{j_1,j_1}_{t_1,t_0}(\varrho_{\rm B})\right].
	\end{eqnarray}
	In general, this expression cannot be further simplified by using the property $\sum_{j_1}E_{j_1}=\oper$, like we did in the previous equation, because of the dependence on $j_1$-dependent terms of the time-dependent dephasing factor. However:
	\begin{proposition}\label{prop:2classicality_maxmix}
		Consider the process with initial state $\rho_{t_0}=\oper/d$ and any measurement PVM $\{P_x\}_{x}$. Then the process is $2$-classical in the two following cases:
		\begin{itemize}
			\item if the process is Markovian;
			\item if the Hamiltonians $\{H_j\}_j$ commute.
		\end{itemize}
	\end{proposition}
	\begin{proof}
		Recalling Prop.~\ref{prop:markommute}, in both cases listed above we have
		\begin{equation}
			\tr\left[\mathcal{U}^{j_2,\ell_2}_{t_2,t_1}\mathcal{U}^{j_1,j_1}_{t_1,t_0}(\varrho_{\rm B})\right]=\tr\left[\mathcal{U}^{j_2,\ell_2}_{t_2,t_1}(\varrho_{\rm B})\right].
		\end{equation}
		Therefore, in both cases, Eq.~\eqref{eq:p2_maxmix} yields
		\begin{eqnarray}\label{eq:p2_maxmix_2}
			\mathbb{P}_2(x_2,t_2;x_1,t_1)&=&\frac{1}{d}\sum_{j_2,\ell_2}\sum_{j_1}\tr\left[P_{x_2}E_{j_2}P_{x_1}E_{j_1}P_{x_1}E_{\ell_2}P_{x_2}\right]\tr\left[\mathcal{U}^{j_2,\ell_2}_{t_2,t_1}(\varrho_{\rm B})\right]\nonumber\\
			&=&\frac{1}{d}\sum_{j_2,\ell_2}\tr\left[P_{x_2}E_{j_2}P_{x_1}E_{\ell_2}P_{x_2}\right]\tr\left[\mathcal{U}^{j_2,\ell_2}_{t_2,t_1}(\varrho_{\rm B})\right],
		\end{eqnarray}
		and marginalizing
		\begin{eqnarray}
			\sum_{x_1}\mathbb{P}_2(x_2,t_2;x_1,t_1)&=&\frac{1}{d}\sum_{j_2,\ell_2}\tr\left[P_{x_2}E_{j_2}E_{\ell_2}P_{x_2}\right]\tr\left[\mathcal{U}^{j_2,\ell_2}_{t_2,t_1}(\varrho_{\rm B})\right]\nonumber\\
			&=&\frac{1}{d}\sum_{j_2}\tr\left[P_{x_2}E_{j_2}P_{x_2}\right]=\frac{1}{d}\tr\left[P_{x_2}\right],
		\end{eqnarray}
		which is the same expression as in Eq.~\eqref{eq:p1_maxmix}.
	\end{proof}
	The above result generalizes what was already observed in the qubit case (cf.~Prop.~\ref{prop:qubit2}): choosing the initial state to be maximally mixed always ensures $2$-classicality, no matter how the system is measured. Incidentally, for non-Markovian processes, this proposition provides an operational way of checking whether the dephasing process is generated by a commutative family of Hamiltonians---a property that \textit{cannot} be checked by simply looking at the reduced dynamics of the process, no matter the choice of the initial state. Any violation of this condition signals the presence of non-commuting blocks.
	
	Again, the argument does not work for higher orders: similar computations can be carried out when marginalizing on the \textit{last} variable $x_1$, but not for the remaining ones. For example:
	\small
	\begin{eqnarray}
		&& \mathbb{P}_3(x_3,t_3;x_2,t_2;x_1,t_1) \nonumber\\
		&& = \frac{1}{d}\sum_{j_3,\ell_3}\sum_{j_2,\ell_2}\sum_{j_1,\ell_1}\tr\left[P_{x_3}E_{j_3}P_{x_2}E_{j_2}P_{x_1}E_{j_1}E_{\ell_1}P_{x_1}E_{\ell_2}P_{x_2}E_{\ell_3}P_{x_3}\right]\tr\left[\mathcal{U}^{j_3,\ell_3}_{t_3,t_2}\mathcal{U}^{j_2,\ell_2}_{t_2,t_1}\mathcal{U}^{j_1,\ell_1}_{t_1,t_0}(\varrho_{\rm B})\right]\nonumber\\
		&& = \frac{1}{d}\sum_{j_3,\ell_3}\sum_{j_2,\ell_2}\sum_{j_1}\tr\left[P_{x_3}E_{j_3}P_{x_2}E_{j_2}P_{x_1}E_{j_1}P_{x_1}E_{\ell_2}P_{x_2}E_{\ell_3}P_{x_3}\right]\tr\left[\mathcal{U}^{j_3,\ell_3}_{t_3,t_2}\mathcal{U}^{j_2,\ell_2}_{t_2,t_1}\mathcal{U}^{j_1,j_1}_{t_1,t_0}(\varrho_{\rm B})\right]\nonumber\\
		&& = \frac{1}{d}\sum_{j_3,\ell_3}\sum_{j_2,\ell_2}\sum_{j_1}\tr\left[P_{x_3}E_{j_3}P_{x_2}E_{j_2}P_{x_1}E_{j_1}P_{x_1}E_{\ell_2}P_{x_2}E_{\ell_3}P_{x_3}\right]\tr\left[\mathcal{U}^{j_3,\ell_3}_{t_3,t_2}\mathcal{U}^{j_2,\ell_2}_{t_2,t_1}(\varrho_{\rm B})\right]\nonumber\\	
		&& = \frac{1}{d}\sum_{j_3,\ell_3}\sum_{j_2,\ell_2}\tr\left[P_{x_3}E_{j_3}P_{x_2}E_{j_2}P_{x_1}E_{\ell_2}P_{x_2}E_{\ell_3}P_{x_3}\right]\tr\left[\mathcal{U}^{j_3,\ell_3}_{t_3,t_2}\mathcal{U}^{j_2,\ell_2}_{t_2,t_1}(\varrho_{\rm B})\right].
	\end{eqnarray}
	\normalsize
	Marginalizing over $x_1$ correctly yields the $2$-time probability distribution, but marginalizing over $x_2$ does not. Again, we can thus construct examples of $2$-classical processes which fail to be classical at higher orders.
	
	\section{Classical Markovian dephasing processes}\label{sec:4}
	
	We shall now reprise the scheme presented in Section~\ref{subsec:mubs}, i.e.~a dephasing process on a basis $\{\ket{\e_j}\}_{j=0,\dots,d-1}$ repeatedly measured in an unbiased basis as in Eq.~\eqref{eq:mubs}. We will now focus on the \textit{Markovian} case, i.e., by Prop.~\ref{prop:markov}, we will assume that the dephasing tensor satisfies Eq.~\eqref{eq:dephasing_regr}.
	
	In this case, recalling the discussion in Section~\ref{sec:1}, classicality is equivalent to the fulfillment of the two Chapman--Kolmogorov equalities, the first one being satisfied because of Prop.~\ref{prop:2classicality_mubs}, the process will be \textit{classical} at any order if and only if the following marginalizing condition for the $3$-time probability distribution holds:
	\begin{equation}\label{eq:markovmarginal}
		\mathbb{P}_2(x_3,t_3;x_1,t_1)=\sum_{x_2=0}^{d-1}\mathbb{P}_3(x_3,t_3;x_2,t_2;x_1,t_1).
	\end{equation}
	We will show that this happens in a particular case, starting from the qubit scenario and then discussing the general scenario.
	
	\subsection{Qubit case}
	
	Recall that, in the qubit case, the unbiased bases as defined by Eq.~\eqref{eq:mubs} are given by Eqs.~\eqref{eq:mubs_qubit}--\eqref{eq:mubs_qubit2}.
	
	\begin{proposition}\label{prop:markov_qubit}
		Let $\dim\hilb_{\rm S}=2$ and consider a Hamiltonian $\textbf{H}=\sum_j\ketbra{\e_j}{\e_j}\otimes H_j$ such that the corresponding dephasing process is Markovian; consider an initial state $\rho_{t_0}$ diagonal in the dephasing basis $\{\ket{\e_j}\}_{j=0,1}$, and choose as measurement basis any of the unbiased bases $\{\ket{\m_x(\phi)}\}_{x=0,1}$ in Eq.~\eqref{eq:mubs_qubit}. Then the process is classical if and only if the dephasing function,
		\begin{equation}\label{eq:phi}
			\varphi_{t,s}=\tr\left[\mathcal{U}^{0,1}_{t,s}(\varrho_{\rm B})\right],
		\end{equation}
		is real-valued.
	\end{proposition}
	\begin{proof}
		As discussed, classicality in this case is ensured if and only if Eq.~\eqref{eq:markovmarginal} holds. For any initial state $\rho_{t_0}$ which is diagonal in the dephasing basis, we get\small
		\begin{eqnarray}\label{eq:markovqubit_p2}
			\mathbb{P}_2(x_3,t_3;x_1,t_1)&=&\frac{1}{8}\sum_{j_3,\ell_3}\sum_{j_1}(-1)^{-(j_3-\ell_3)x_3}(-1)^{(j_3-\ell_3)x_1}\braket{j_1|\rho_{t_0}|j_1}\tr\left[\mathcal{U}^{j_3,\ell_3}_{t_3,t_1}(\varrho_{\rm B})\right]\nonumber\\
			&=&\frac{1}{8}\sum_{j_3,\ell_3}(-1)^{-(j_3-\ell_3)(x_3-x_1)}\tr\left[\mathcal{U}^{j_3,\ell_3}_{t_3,t_1}(\varrho_{\rm B})\right],
		\end{eqnarray}\normalsize
		while\small
		\begin{eqnarray}
			\mathbb{P}_3(x_3,t_3;x_2,t_2;x_1,t_1)&=&\frac{1}{32}\sum_{j_3,\ell_3}\sum_{j_2,\ell_2}\sum_{j_1}(-1)^{-(j_3-\ell_3)x_3}(-1)^{[(j_3-\ell_3)-(j_2-\ell_2)]x_2}(-1)^{(j_2-\ell_2)x_1}\nonumber\\
			&&\times\braket{j_1|\rho_{t_0}|j_1}\tr\left[\mathcal{U}^{j_3,\ell_3}_{t_3,t_2}(\varrho_{\rm B})\right]\tr\left[\mathcal{U}^{j_2,\ell_2}_{t_2,t_1}(\varrho_{\rm B})\right]\nonumber\\
			&=&\frac{1}{32}\sum_{j_3,\ell_3}\sum_{j_2,\ell_2}(-1)^{-(j_3-\ell_3)x_3}(-1)^{[(j_3-\ell_3)-(j_2-\ell_2)]x_2}(-1)^{(j_2-\ell_2)x_1}\nonumber\\
			&&\times\tr\left[\mathcal{U}^{j_3,\ell_3}_{t_3,t_2}(\varrho_{\rm B})\right]\tr\left[\mathcal{U}^{j_2,\ell_2}_{t_2,t_1}(\varrho_{\rm B})\right],
		\end{eqnarray}
		\normalsize whence \small
		\begin{eqnarray}
			\sum_{x_2}	\mathbb{P}_3(x_3,t_3;x_2,t_2;x_1,t_1)&=&\frac{1}{32}\sum_{j_3,\ell_3}\sum_{j_2,\ell_2}\left(\sum_{x_2}(-1)^{[(j_3-\ell_3)-(j_2-\ell_2)]x_2}\right)(-1)^{-(j_3-\ell_3)x_3}(-1)^{(j_2-\ell_2)x_1}\nonumber\\
			&&\times\tr\left[\mathcal{U}^{j_3,\ell_3}_{t_3,t_2}(\varrho_{\rm B})\right]\tr\left[\mathcal{U}^{j_2,\ell_2}_{t_2,t_1}(\varrho_{\rm B})\right].
		\end{eqnarray}\normalsize
		The sum between parentheses is equal to $d=2$ when $(j_3-\ell_3)-(j_2-\ell_2)$ is an integer multiple of $d=2$, and zero otherwise. The quantity $(j_3-\ell_3)-(j_2-\ell_2)$ can take values at most between $\pm2$, whence the sum equals $2$ whenever this quantity equals either $-2,0,2$, and is zero otherwise. That is,
		\begin{eqnarray}\label{eq:eccedelta}
			\sum_{x_2}(-1)^{[(j_3-\ell_3)-(j_2-\ell_2)]x_2}&=&2\left(\delta_{j_3-\ell_3,j_2-\ell_2}+\delta_{j_3-\ell_3,j_2-\ell_2-2}+\delta_{j_3-\ell_3,j_2-\ell_2+2}\right)\nonumber\\
			&=&2\sum_{k=0,\pm1}\delta_{j_3-\ell_3,j_2-\ell_2+2k}.
		\end{eqnarray}
		We end up with
		\begin{eqnarray}\label{eq:markovqubit_p3}
			\sum_{x_2}\mathbb{P}_3(x_3,t_3;x_2,t_2;x_1,t_1)&=&\frac{1}{16}\sum_{j_3,\ell_3}(-1)^{-(j_3-\ell_3)(x_3-x_1)}\tr\left[\mathcal{U}^{j_3,\ell_3}_{t_3,t_2}(\varrho_{\rm B})\right]\nonumber\\&&\times\sum_{j_2,\ell_2}\left(\sum_{k=0,\pm1}\delta_{j_3-\ell_3,j_2-\ell_2+2k}\right)\tr\left[\mathcal{U}^{j_2,\ell_2}_{t_2,t_1}(\varrho_{\rm B})\right].
		\end{eqnarray}
		We can now compute the difference between the two-time probability and the marginalized three-time probability: noting that all terms with $j_3=\ell_3$ cancel out, and using the shorthand~\eqref{eq:phi} and the Markovianity property $\varphi_{t_3,t_1}=\varphi_{t_3,t_2}\varphi_{t_2,t_1}$ (cf.~Prop.~\ref{prop:markov}), we get\small
		\begin{eqnarray}
			&&16(-1)^{x_3-x_1}\left[\mathbb{P}_2(x_3,t_2;x_1,t_1)-\sum_{x_2}\mathbb{P}_3(x_3,t_3;x_2,t_2;x_1,t_1)\right]\nonumber\\
			&=&2\varphi_{t_3,t_1}+2\varphi_{t_3,t_1}^*-\varphi_{t_3,t_2}\varphi_{t_2,t_1}-\varphi_{t_3,t_2}\varphi_{t_2,t_1}^*-\varphi_{t_3,t_2}^*\varphi_{t_2,t_1}-\varphi_{t_3,t_2}^*\varphi_{t_2,t_1}^*\nonumber\\
			&=&\varphi_{t_3,t_2}\varphi_{t_2,t_1}+\varphi_{t_3,t_2}^*\varphi_{t_2,t_1}^*-\varphi_{t_3,t_2}\varphi_{t_2,t_1}^*-\varphi_{t_3,t_2}^*\varphi_{t_2,t_1}\nonumber\\
			&=&-4\Im\varphi_{t_3,t_2}\Im\varphi_{t_2,t_1},
		\end{eqnarray}\normalsize
		therefore this quantity vanishes if and only if the phase is real, thus completing the proof.
	\end{proof}
	The difference between the $2$-time correlation function and the $3$-time correlation function marginalized on its central variable is thus directly related to the magnitude of the imaginary part of the dephasing function $\varphi_{t,s}$. In particular, if the global Hamiltonian is time-independent, then
	\begin{equation}
		\varphi_{t,s}=\e^{-(t-s)\left(\frac{\gamma}{2}+\i\varepsilon\right)},
	\end{equation}
	and thus\small
	\begin{eqnarray}
		\mathbb{P}_2(x_3,t_2;x_1,t_1)-\sum_{x_2}\mathbb{P}_3(x_3,t_3;x_2,t_2;x_1,t_1)=\frac{1}{8}(-1)^{x_3-x_1}\e^{-\frac{\gamma}{2}(t_3-t_1)}\sin\left[\varepsilon(t_3-t_2)\right]\sin\left[\varepsilon(t_2-t_1)\right],
	\end{eqnarray}\normalsize
	whence classicality holds if and only if $\varepsilon=0$. We remark that, in the case in which $\textbf{H}$ is a Markovian dephasing-type spin--boson model (i.e.~with the spin--boson coupling being mediated by a flat form factor, cf.~\cite{Davide-2}), $\varepsilon$ corresponds to the excitation energy of the qubit this means that the process is classical in the limit in which the excitation energy vanishes---that is, the difference between the energy levels of the qubit is negligible.
	
	Finally, for completeness let us compute the whole $2$-time and $3$-time statistics in the Markovian qubit scenario under the assumption of a real dephasing function. In this case, we have explicitly
	\begin{eqnarray}
		\mathbb{P}_2(0,t_3;0,t_1) &=& \mathbb{P}_2(1,t_3;1,t_1) =  \frac 14 (1 + \varphi_{t_3,t_1} ) ,\nonumber \\
		\mathbb{P}_2(0,t_3;1,t_1) &=& \mathbb{P}_2(1,t_3;0,t_1)= \frac 14 (1 - \varphi_{t_3,t_1} ) ,
	\end{eqnarray}	
	\begin{eqnarray}
		\mathbb{P}_3(0,t_3;0,t_2;0,t_1) &=& \mathbb{P}_3(1,t_3;1,t_2;1,t_1) =\frac 18 (1 + \varphi_{t_3,t_2} + \varphi_{t_2,t_1} + \varphi_{t_3,t_1} ) ,\nonumber\\
		\mathbb{P}_3(0,t_3;1,t_2;0,t_1) &=& \mathbb{P}_3(1,t_3;0,t_2;1,t_1)=\frac 18 (1 - \varphi_{t_3,t_2} - \varphi_{t_2,t_1} + \varphi_{t_3,t_1} ) ,\nonumber\\
		\mathbb{P}_3(0,t_3;0,t_2;1,t_1) &=& 		\mathbb{P}_3(1,t_3;1,t_2;0,t_1) =\frac 18 (1 + \varphi_{t_3,t_2} - \varphi_{t_2,t_1} - \varphi_{t_3,t_1} ) ,\\
		\mathbb{P}_3(1,t_3;0,t_2;0,t_1) &=& \mathbb{P}_3(0,t_3;1,t_2;1,t_1) = \frac 18 (1 - \varphi_{t_3,t_2} + \varphi_{t_2,t_1} - \varphi_{t_3,t_1} ) ,\nonumber
	\end{eqnarray}
	and one immediately checks that all consistency conditions are satisfied.
	
	\subsection{General case}
	A generalization of Prop.~\ref{prop:markov_qubit} holds in the general case, $\dim\hilb_{\rm S}=d$, by properly upgrading the requirement of a real dephasing function to the $d$-dimensional scenario. We shall start with a simple preliminary lemma:
	\begin{lemma}\label{lemma}
		Let $d\in\mathbb{N}$ and $j,\ell=0,\dots,d-1$. Then, given $h\in\{\pm1,\pm2,\dots,\pm(d-1)\}$, the following equality holds:
		\begin{equation}\label{eq:delta}
			\sum_{j\neq\ell}\sum_{k=0,\pm1}\delta_{j-\ell,h+kd}=d,
		\end{equation}
	with $\delta_{j\ell}$ being the Kronecker delta.
	\end{lemma}
	\begin{proof}
		Given $k\in\mathbb{N}$, the cardinality $\#$ of the solutions of the equation $j-\ell=k$ can be simply shown to be
		\begin{equation}\label{eq:number}
			\#\left\{(j,\ell)\in\{1,\dots,d\}^2:\;j-\ell=k\right\}=\begin{cases}
				d-|k|,&|k|=0,1,\dots,d-1;\\
				0,&|k|=d,d+1,\dots,
			\end{cases}
		\end{equation}
		obviously with all solutions being characterized by $j=\ell$ iff $k=0$ and $j\neq\ell$ iff $k\neq0$.
		
		Now, the quantity in Eq.~\eqref{eq:delta} corresponds to the cumulative number of solutions of the three equations
		\begin{equation}
			j-\ell=h,\qquad j-\ell=h-d,\qquad j-\ell=h+d.
		\end{equation}
		Suppose e.g.~$h>0$. Then, by Eq.~\eqref{eq:number}, the first equation admits $d-h$ solutions, the second one admits $d-|h-d|=h$ solutions, and the third one does not admit solutions; whence Eq.~\eqref{eq:delta} is proven. The same holds for $h<0$.
	\end{proof}
	
	\begin{proposition}\label{prop:markov_qudit}
		Let $\dim\hilb_{\rm S}=d$ and consider a Hamiltonian $\textbf{H}=\sum_j\ketbra{\e_j}{\e_j}\otimes H_j$ such that the corresponding dephasing process is Markovian; consider an initial state $\rho_{t_0}$ diagonal in the dephasing basis $\{\ket{\e_j}\}_{j=0,\dots,d-1}$, and choose as measurement basis $\{\ket{\m_x(\boldsymbol\phi)}\}_{x=0,\dots,d-1}$ as in Eq.~\eqref{eq:mubs}. Assume that the following equality holds:
		\begin{equation}\label{eq:matrix}
			\tr\;\mathcal{U}^{j,\ell}_{t,s}(\varrho_{\rm B})=\begin{cases}
				1,&j=\ell;\\
				\varphi_{t,s},&j\neq\ell
			\end{cases}
		\end{equation}
		for some real-valued function $\varphi_{t,s}$. Then the process is classical.
	\end{proposition}
	In words, in the larger-dimensional scenario classicality holds under the following, quite restrictive, requirement: all off-diagonal elements of the density operator in the dephasing basis must evolve in the same way---in particular, with real-valued dephasing factors.
	\begin{proof}
		Assume Eq.~\eqref{eq:matrix}. The two-time probability is given by
		\begin{equation}
			\mathbb{P}_2(x_3,t_3;x_1,t_1)=\frac{1}{d^3}\left[d+\varphi_{t_3,t_1}\sum_{j_3\neq\ell_3} \omega^{-(j_3-\ell_3)(x_3-x_1)}\right].
		\end{equation}
		As for the three-time probability, notice that Eq.~\eqref{eq:eccedelta} generalizes as
		\begin{eqnarray}\label{eq:eccedelta2}
			\sum_{x_2}(-1)^{[(j_3-\ell_3)-(j_2-\ell_2)]x_2}=d\sum_{k=0,\pm1}\delta_{j_3-\ell_3,j_2-\ell_2+kd},
		\end{eqnarray}
		whence, by following analogous computations as in the qubit case,\small
		\begin{eqnarray}
			\sum_{x_3=0}^{d-1}\mathbb{P}_3(x_3,t_3;x_2,t_2;x_1,t_1)
			&=&\!\frac{1}{d^4}\!\left[d^2\!+\!\varphi_{t_3,t_1}\varphi_{t_2,t_1}\!\sum_{j_3\neq\ell_3}\! \omega^{-(j_3-\ell_3)(x_3-x_1)}\sum_{j_2\neq\ell_2}\sum_{k=0,\pm1}\delta_{j_3-\ell_3,j_2-\ell_2+kd}\right]\nonumber\\
			&=&\!\frac{1}{d^4}\!\left[d^2+d\,\varphi_{t_3,t_1}\varphi_{t_2,t_1}\!\sum_{j_3\neq\ell_3}\! \omega^{-(j_3-\ell_3)(x_3-x_1)}\right]\nonumber\\
			&=&\mathbb{P}_2(x_3,t_3;x_1,t_1),
		\end{eqnarray}\normalsize
		where we have applied Lemma~\ref{lemma}.
	\end{proof}
	
	\section{Conclusions}\label{sec:5}
	
	We have analyzed the multitime statistics associated with generic quantum dephasing processes, finding scenarios in which the family of probability distributions is either classical up to a finite order or exhibits a full-fledged classical character. In particular, excluding the case in which the system is probed via measurements compatible with the dephasing basis, a classical behavior emerges for a wide class of measurement bases which are unbiased with respect to the dephasing basis; in particular, in the Markovian case sufficient conditions for classicality at any order can be found. As a consequence of this property, the following dichotomy can be observed; classicality can be achieved in two very different scenarios: either the measurement system is \textit{fully compatible} with the dephasing basis, or it is \textit{fully incompatible} with it.
	
	Besides, a simple and interesting link between classicality and the commutativity of the Hamiltonians which generate the dephasing dynamics: when preparing the system in the maximally mixed state, the corresponding process will be $2$-classical, regardless the choice of the measurement basis, whenever these Hamiltonians commute; in the non-Markovian case, any violation of classicality at the level of $2$-time correlation functions will be a witness of noncommutativity. More generally, while the commutativity between the Hamiltonians has no direct influence on the reduced dynamics of the system, it \textit{does} have a potentially crucial role when taking into account the role of interventions on the system, since it implies a highly nontrivial constrain on the $n$-time dephasing tensor which affects the property of the multitime statistics of the process. A more systematic study of the role of commutativity in the properties of dephasing processes is highly desirable.
	
	These results reveal a key difference between Markovian and non-Markovian dephasing processes. In the Markovian case, a genuinely time-dependent dephasing process can be fully classical, in the sense adopted in the present paper, modulo a proper choice of initial state and a suitable choice of measurement basis, i.e.~any of the unbiased bases in Eq.~\eqref{eq:mubs}. In the qubit case, this choice of basis is indeed necessary other than sufficient to obtain classicality; we conjecture that the same holds for the qu$d$it case. Instead, in the non-Markovian case the same choices of initial state and measurement basis do not guarantee classicality: only $2$-classicality is guaranteed. We may conjecture that, in the non-Markovian case, a pure dephasing process can only be classical, possibly modulo very special choices of initial state, if one performs a fully compatible measurement with respect to the dephasing basis---that is, no genuinely time-dependent dephasing process can ultimately ``hide'' its quantumness.
	
	In this regard, we recall that, in~\cite{Andrea-1,Andrea-2}, an interesting characterization of classical Markovian processes was found: a Markovian process was found to be classical if and only if it can be reproduced by means of an initial system $\tilde{\rho}_{t_0}$ which is diagonal in the \textit{measurement} basis and a family of propagators $\tilde{\Lambda}_{t,s}$ satisfying the non-coherence-generating-and-detecting (NCGD) property:	
	\begin{equation}\label{NCGD}
		\Delta \circ \tilde{\Lambda}_{t_3,t_2} \circ \Delta \circ \tilde{\Lambda}_{t_2,t_1} \circ \Delta = \Delta \circ \tilde{\Lambda}_{t_3,t_1} \circ \Delta ,
	\end{equation}
	where $\Delta = \sum_x P_x(\cdot) P_x$ denotes the completely dephasing channel with respect to the \textit{measurement} basis. Importantly, both $\tilde{\rho}_{t_0}$ and $\tilde{\Lambda}_{t,s}$ are generally different from the ``physical'' initial state and propagator of the process. The Markovian scenario considered in this paper required a physical initial state which is diagonal in the dephasing basis; however, the family of physical propagators ${\Lambda}_{t,s}$ \textit{does} satisfy Eq.~\eqref{NCGD}. Indeed, one shows that
	\begin{equation}
		\Delta \circ {\Lambda}_{t,s} \circ \Delta = {\Lambda}_{t,s} \circ \Delta,
	\end{equation}
	which implies Eq.~\eqref{NCGD}. Note that, by measuring the system with respect to an arbitrary MUB, one creates quantum coherences with respect to the dephasing basis in an extreme manner: however, such generated coherences cannot be detected at a later time by means of the employed measurement basis.

	We recall that, due to the seminal paper by Leggett and Garg~\cite{LG1,LG2}, a violation of classicality is
	understood as measurement invasiveness (impossibility to perform a  measurement without altering the state of the system). Recently \cite{Andrea-exp} it was experimentally verified that the violation of the Kolmogorov conditions is proportional to the amount of quantum coherence (of the measured observable). Such a connection between coherence and non-classicality was probed in a time-multiplexed optical quantum walk~\cite{Andrea-exp} (cf.~also~\cite{LG2} for other experiments testing the violation of Leggett--Garg inequalities).
	
	To conclude, reprising the discussion in Section~\ref{sec:1}, it will be useful to compare our results to the ones in~\cite{Lukasz}. The authors, while still taking into account dephasing phenomena, analyze the different scenario in which the system is reset at its initial state after each measurement. With this approach, a simple and elegant characterization of classicality is found: using our nomenclature, the multitime statistics corresponding to a dephasing process which is repeatedly measured \textit{and} reset to a given initial state is classical at any order \textit{if and only if} the environment Hamiltonians $\{H_j\}_j$ commute, provided that a nondegeneracy assumption is satisfied;  This scenario shows important difference from the results obtained in the present paper, in which the commutativity of $\{H_j\}_j$, while playing a distinguished role (cf.~the results in Section~\ref{subsec:maxmix}), is neither sufficient not necessary for classicality. Consequently, the statistics with and without re-preparation of the system in the initial state behave in a starkly different way when it comes to classicality, despite the fact that the two statistics can be mapped into each other~\cite{Lukasz2}.
	
	Finally, as a further connection between the two approaches, we observe that the authors of~\cite{Lukasz} examine the situation in which the system is measured with the eigenbasis of the Pauli operators $\sigma_x$, $\sigma_y$, and propose the use of MUBs as a higher-generalization of such results---the very same scenario considered in this work, in which MUBs do indeed play a decisive role.
	
	This brief discussion strongly hints at the existence of a deeper connection between the results of~\cite{Lukasz} and those presented in this work; this shall be examined in the future.
	
	\section*{Acknowledgments}
	
	We acknowledge valuable discussions with Fattah Sakuldee and Lukasz Cywi\'nski. DL was partially supported by Istituto Nazionale di Fisica Nucleare (INFN) through the project “QUANTUM” and by the Italian National Group of Mathematical Physics (GNFM--INdAM), and acknowledges support by MIUR via PRIN 2017 (Progetto di Ricerca di Interesse Nazionale), project QUSHIP (2017SRNBRK); he also thanks the Institute of Physics at the Nicolaus University in Toru\'n for its hospitality. DC was supported by the Polish National Science Center project No. 2018/30/A/ST2/00837.

\end{document}